\documentclass[11pt, journal, draftclsnofoot, onecolumn]{IEEEtran}
\usepackage{mathtools,enumitem,color,subfigure,bbm}
\usepackage{graphicx,tabularx,array}
\usepackage{hyperref}

\usepackage{amsthm,mathrsfs}
\usepackage{makecell}
\usepackage{multirow}
\usepackage{txfonts}
\usepackage[T1]{fontenc}
\usepackage{cite}

\usepackage{lipsum,geometry}

\definecolor{green}{rgb}{0.6627,0.8196,0.5568}
\definecolor{blue}{rgb}{0.6875,0.8750,0.8984}
\definecolor{purple}{rgb}{ 0.7647,    0.6078,    0.8824}

\newcommand{\discardpages}[1]{
  \xdef\discard@pages{#1}
  \AtBeginShipout{
    \renewcommand*{\do}[1]{
      \ifnum\value{page}=##1\relax%
        \AtBeginShipoutDiscard
        \gdef\do####1{}
      \fi%
    }%
    \expandafter\docsvlist\expandafter{\discard@pages}
  }%
}
\newif\ifkeeppage
\newcommand{\keeppages}[1]{
  \xdef\keep@pages{#1}
  \AtBeginShipout{
    \keeppagefalse%
    \renewcommand*{\do}[1]{
      \ifnum\value{page}=##1\relax%
        \keeppagetrue
        \gdef\do####1{}
      \fi%
    }%
    \expandafter\docsvlist\expandafter{\keep@pages}
    \ifkeeppage\else\AtBeginShipoutDiscard\fi
  }%
}

\newtheorem{Theorem}{Theorem}

\newtheorem{Remark}{Remark}

\begin{document}
%

\title{MIMO Systems with One-bit ADCs: Capacity Gains using Nonlinear Analog Operations}

\author{

\IEEEauthorblockN{ Farhad Shirani$^\dagger$, Hamidreza Aghasi$\ddagger$}
\\\IEEEauthorblockA{$^\dagger$ North Dakota State University, $^\ddagger$ University of California, Irvine,
\\Email: f.shiranichaharsoogh@ndsu.edu, haghasi@uci.edu}
}


%


\maketitle

 \begin{abstract}
Analog to Digital Converters (ADCs) are a major contributor to the energy consumption on the receiver side of millimeter-wave multiple-input multiple-output (MIMO) systems with large antenna arrays. Consequently, there has been significant interest in using low-resolution ADCs along with hybrid beam-forming at MIMO receivers for energy efficiency. However, decreasing the ADC resolution results in performance loss --- in terms of achievable rates --- due to increased quantization error. In this work, we study the application of practically implementable nonlinear analog operations, prior to sampling and quantization at the ADCs, as a way to mitigate the aforementioned rate-loss.
A receiver architecture consisting of linear analog combiners, implementable nonlinear analog operators, and one-bit threshold ADCs is designed.
The fundamental information theoretic performance limits of the resulting communication system, in terms of achievable rates,  are investigated under various assumptions on the set of implementable nonlinear analog functions. In order to justify the feasibility of the nonlinear operations in the proposed receiver architecture, an analog circuit is introduced, and circuit simulations exhibiting the generation of the desired nonlinear analog operations are provided. 
\end{abstract}


%
\IEEEpeerreviewmaketitle

\section{Introduction}
In order to satisfy the ever-growing demand for higher data-rates and bandwidth, the emerging wireless networks operate in frequencies {above 6 GHz}  especially the millimeter wave (mm-wave) bands. The high carrier frequencies used in mm-wave systems allow for larger channel bandwidths
compared to lower-frequency systems. For instance, in conventional protocols such as LTE, the bandwidth is between 1.4 MHz and 20 MHz \cite{LTE1}, whereas in 
microwave WiFi standards such as IEEE 802.11ad the bandwidth is 2.16 GHz \cite{802.11ad}, and in mm-wave cellular applications, bandwidth of 500 MHz or more has been considered \cite{bw2}. The inherent high isotropic path loss and sensitivity to blockages at high frequencies
pose challenges in supporting high capacity and mobility \cite{rappaport2015millimeter}. As an example, 
Friis’ Law states that the isotropic path loss in free-space propagation is
inversely proportional to the wavelength squared \cite{heath2016overview}. In order to mitigate the path loss, mm-wave systems leverage narrow-beams, by using large antenna arrays at both base stations (BS) and user-ends (UE). For instance, fifth-generation (5G) wireless networks  envision hundreds of antennas at the BS and in excess of ten antennas at the UE \cite{hong2014study}. 
In conventional multiple-input multiple-output (MIMO) systems with digital beamforming, each antenna input/output is digitized separately \cite{DigBF1}. This requires each receiver antenna to be connected to a dedicated analog to digital converter (ADC). Since mm-wave systems use large arrays of antennas --- to mitigate the propagation losses due to small carrier wavelength --- digital beamforming in mm-wave systems requires a large number of ADCs which are a significant source of power consumption in MIMO receivers \cite{heath2016overview,mendez2015channel}.
Another contributing factor to the high energy demands in mm-wave ADC modules is the large channel bandwidth used in these applications. In theory, the power consumption of an ADC grows linearly in bandwidth, and the rate of increase is even more significant in practical implementations due to the excessive loss associated with the passive components at higher frequencies \cite{BR,ADCpower,razaviADC}. Consequently, the massive number of receiver antennas and large channel bandwidth result in a substantial increase in ADC power consumption in mm-wave MIMO systems. Furthermore, in standard ADC design, power consumption is proportional to the number of quantization bins and hence grows exponentially in the number of output bits \cite{walden1999analog}. This   limits the resolution of the ADCs due to power budget restrictions.

Hybrid beam-forming with low-resolution ADCs has been proposed as a way to mitigate the high energy cost of ADCs by reducing the number of converters and their resolutions. To elaborate, under hybrid beam-forming, the receiver terminals in MIMO systems use analog beam-formers to linearly combine the large number of signals at the receiver antennas and feed them to a small set of low-resolution ADCs.
  There has been a large body of work on the design of energy-efficient transceiver architectures and coding strategies using hybrid beam-forming with a small number of ADCs for communication in mm-wave MIMO systems \cite{molisch2017hybrid,heath2016overview,alkhateeb2014mimo,nossek2006capacity,abbasISIT2018,rini2017general,mezghani2009transmit,khalili2020throughput,dutta2020capacity}. The communication setup considered in these works is discussed briefly in Section \ref{sec:form}.

In this work,  we consider the use of nonlinear analog operations as a way to mitigate the rate-loss due to the use of low-resolution ADCs. 
To explain the aforementioned rate gains, let us consider a simple single-input single-output (SISO) scenario operating in the high signal-to-noise ratio (SNR) regime, i.e. $Y\approx X$. Assume that the receiver is equipped with two one-bit threshold ADCs. Then, as shown in Figure \ref{fig:1}(a), it can receive at most three different messages per channel-use by performing two threshold comparisons, e.g. comparisons with threshold zero $Y\lessgtr 0$ and threshold one $Y \lessgtr 1$ , hence achieving a rate of $R=\log{3}$ bits/channel-use. Alternatively, if the receiver has access to the second power, $Y^2$, of the channel output, then it can use the two comparators $Y\lessgtr 0$ and $Y^2 \lessgtr 1$ as shown in Figure \ref{fig:1}(b) to achieve $R=2$ bits/channel-use, hence improving performance. We investigate the set of achievable regions under general assumptions on the number of transmit antennas, receive antennas, one-bit ADCs, channel SNR, and the values of $k$ in $Y^k$ which can be produced using analog circuits. We provide several communication strategies and derive the resulting achievable regions in various scenarios described in Section \ref{sec:form}.  
To justify the feasibility of the nonlinear analog operations studied in this work, we show through simulation of a circuit whose design is explained in Section \ref{sec:cir}, that one can implement ADCs which operate by comparing the $k$th power $Y^k$ of the input to a set of thresholds, where $k>1$, without significant increase in power consumption. At a high level, the proposed circuit operates on the higher harmonics of its input signal to extract the $k$th power of the amplitude, and the value of $k$ is bounded from above due to practical restrictions in circuit design. 
Circuit simulations exhibiting the generation of nonlinear analog operations are provided.  
\begin{figure}[t!]
\centering
\includegraphics[width=0.8\linewidth]{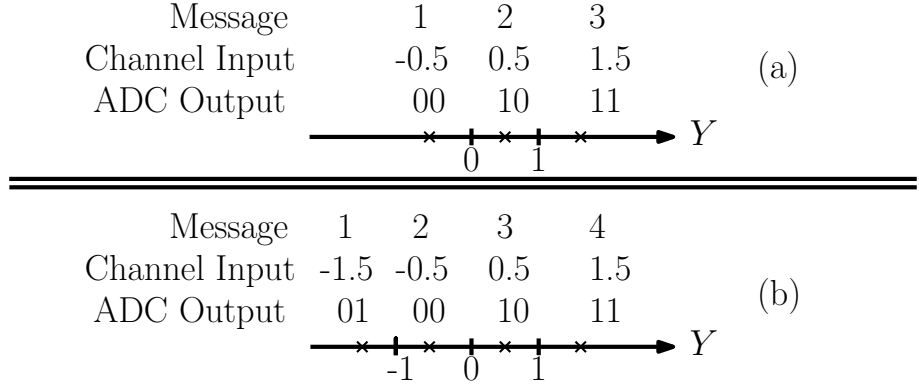}
\caption{(a) the transmitter sends $X$ with amplitudes $-0.5,0.5,1.5$ to send messages $1,2$ and $3$, respectively, and receiver uses two one-bit threshold ADCs $Y\lessgtr 0$ and $Y\lessgtr 1$, and (b)  the transmitter sends $X$ with amplitudes $-1.5, -0.5,0.5,1.5$ to send messages $1,2,3$ and $4$, respectively, and receiver uses two one-bite threshold ADCs $Y\lessgtr 0$ and $Y^2\lessgtr1$.
\vspace{-0.25in}}
\label{fig:1}
\end{figure}

{\em Notation:}
 The set $\{1,2,\cdots, n\}, n\in \mathbb{N}$ is represented by $[n]$. 
The  $n$-length vector $(x_1,x_2,\hdots, x_n)$ is written as $x(1\!\!:\!\!n)$ and $x^n$, interchangeably, and $(x_k,x_{k+1},\cdots,x_n)$ is denoted by $x(k:n)$. The $i$th element is written as $x(i)$ and $x_i$, interchangeably. We write $||\cdot||_2$ to denote the $L_2$-norm. An $n\times m$ matrix is written as $h(1\!\!:\!\!n,1\!\!:\!\!m)=[h_{i,j}]_{i,j\in [n]\times [m]}$,
, its $i$th column is $h(:,i), i\in [m]$, and its $j$th row is $h(j,:), j\in [m]$. We write $\mathbf{x}$ and $\mathbf{h}$ instead of $x(1\!\!:\!\!n)$ and $h(1\!\!:\!\!n,1\!\!:\!\!m)$, respectively, when the dimension is clear from context. Sets are denoted by calligraphic letters such as $\mathcal{X}$, families of sets by sans-serif letters such as $\mathsf{X}$, and collections of families of sets by $\mathscr{X}$. For the region $\mathcal{A}\in \mathbb{R}^n$, the set $\partial\mathcal{A}_k$ denotes its boundary. $\mathbb{B}$ denotes the Borel $\sigma$-field.

\section{System Model}
\label{sec:form}
We consider a MIMO communication channel characterized by the triple $(n_t,n_r, \mathbf{h})$, where $n_t$ is the number of transmitter antennas, $n_r$ is the number of receiver antennas, and $\mathbf{h}\in \mathbb{R}^{n_t\times n_r}$ is the (fixed) channel gain matrix. It is assumed that the transmitter and receiver have prefect knowledge of $\mathbf{h}$.
The channel input and output $(\mathbf{X}, \mathbf{Y})\in \mathbb{R}^{n_t}\times \mathbb{R}^{n_r}$ are related through $
\mathbf{Y}=\mathbf{h}\mathbf{X}+\mathbf{N}$, 
where $\mathbf{N}\in \mathbb{R}^{n_r}$ is a vector of independent and identically distributed Gaussian variables with unit variance and zero mean, and the channel input has average power constraint $P$, i.e. $\frac{1}{n_t}\sum_{i=1}^{n_t}\mathbb{E}(X^2_i)\leq P$.
 \begin{figure}[t]
\centering 
\includegraphics[width=0.8\linewidth]{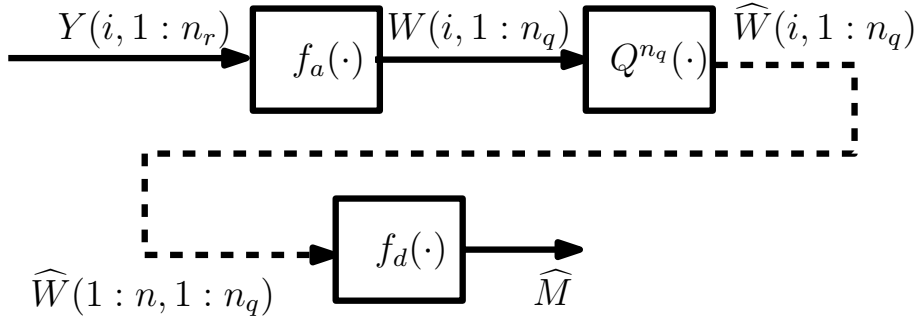}
\caption{The receiver architecture consists of an analog processing module $f_a(\cdot)$, $n_q$ one-bit threshold ADCs $Q^{n_q}(\cdot)$, and a digital processing module $f_d(\cdot)$. The function $f_a(\cdot)$ is restricted to compositions of linear analog combiners and specific  nonlinear analog functions.
}
\vspace{-.25in}
\label{fig:receiver}
\end{figure}

Let the message $M$ be chosen randomly and uniformly from $[\Theta]$, where $\Theta \in \mathbb{N}$. The communication blocklength is $n\in \mathbb{N}$ and  the communication rate is $\frac{1}{n}\log{\Theta}$. The transmitter produces $X(1\!\!:\!\!n,1\!\!:\!\!n_t)= e(M)$, where $e:[\Theta]\to \mathbb{R}^{n\times n_t}$ is the encoding function. At the $i$th channel-use, the vector $X(i,1\!\!:\!\!n_t), i\in [n]$ is transmitted and the receiver receives $Y(i,1\!:\!n_t)=\mathbf{h}X(i,1\!:\!n_t)+N(i,1\!\!:\!\!n_t)$. The receiver produces the message reconstruction $\widehat{M}=d(Y(1\!\!:\!\!n,1\!\!:\!\!n_r))$, where $d: \mathbb{R}^{n\times n_r}\to[\Theta]$ is the decoding function. In this work, we restrict the choice of decoding functions by considering the limitations on number of available one-bit threshold ADCs, $n_q\in \mathbb{N}$, and the set of \textit{implementable analog functions} {$\mathcal{F}_a$} at the receiver. This is described in more detail in the following. 

In its most general form, the receiver (Figure \ref{fig:receiver})  consists of three components: i) an analog processing module captured by $f_a: \mathbb{R}^{n_r}\to \mathbb{R}^{n_q}$ operating on each channel output $Y(i,1:n_r),i\in [n]$, ii) a set of $n_q$ one-bit threshold ADCs with threshold vector $t(1\!\!:\!\!n_q)\in \mathbb{R}^{n_q}$ captured by $Q_{t^{n_q}}^{n_q}:\mathbb{R}^{n_q}\to \{0,1\}^{n_q}$ operating on each output of the analog processing module $W(i,1:n_q)$, and iii) a digital processing module captured by $f_d:\{0,1\}^{n\times n_q}\to [\Theta]$, operating on the block of ADC outputs $\widehat{W}(1:n,1:n_q)$.

After the $i$th channel-use, the analog processing module processes the received signal $Y(i,1\!\!:\!\!n_r)$ in the analog domain and produces $W(i,1\!\!:\!\!n_q)=f_a(Y(i,1\!\!:\!\!n_r)), i\in [n]$. The choice of $f_a(\cdot)$ is restricted to the set of implementable analog functions $\mathcal{F}_a$ and is discussed further in the sequel. The output $W(i,1\!:\!n_q), i\in [n]$ is fed to the one-bit threshold ADCs which produce the discretized vector \[\widehat{W}(i,1\!\!:\!\!n_q)=Q_{t^{n_q}}^{n_q}(W(i,1\!\!:\!\!n_q))=
(W(i,j)\lessgtr t(j),j\in [n_q]),\] where $W(i,j)\lessgtr t(j)$ denotes the indicator function $\mathbbm{1}\{W(i,j)>t(j)\}$. After the $n$th channel-use, the digital processing module produces the message reconstruction $\widehat{M}=f_d(\widehat{W}(1\!\!:\!\!n,1\!\!:\!\!n_q))$. The communication system is characterized by $(n_t,n_r,\mathbf{h}, n_q, \mathcal{F}_a)$, and the transmission system by $(n,\Theta,e,f_a,f_d)$, where $f_a\in \mathcal{F}_a$. Achievability and probability of error are defined in the standard information-theoretic sense. The capacity maximized over all implementable analog functions is denoted by $C_{Q}(n_t,n_r,\mathbf{h},P,n_q, {\mathcal{F}_a})$.

\begin{Remark}
We have considered channels with real-valued inputs with one-bit ADCs used at the receiver. The analysis may be extended to complex-valued inputs and low-resolution ADCs with more than one-bit output length in a straightforward manner.
\end{Remark}
\begin{Remark}
We have considered MIMO systems with low-resolution ADCs under average input power constraints.
Peak power constraints have also been considered \cite{witsenhausen1980some,khalili2018mimo}
under specific restrictions on $\mathcal{F}_a$ such as linearity of the analog processing function.
\end{Remark}

The following sets of implementable analog functions $\mathcal{F}_a$ have been considered in prior works:
\begin{itemize}[leftmargin=*]
    \item \textbf{Scenario I: No Analog Processing \cite{singh2009limits}:} $\mathcal{F}^I_a$ consists of the single (trivial) function $f_a(\mathbf{Y})=\mathbf{Y}$, and we must have $n_r=n_q$, i.e. digital beamforming. It was shown that when the signal is not processed in the analog domain, binary antipodal signaling is optimum in all SNRs.
    \item \textbf{Scenario II: Linear Analog Processing \cite{khalili2021mimo,mo2015capacity, alkhateeb2014mimo}:} $\mathcal{F}^{II}_a$ consists of linear functions $f_a(\mathbf{Y})=\mathbf{V}\mathbf{Y}, \mathbf{V}\in \mathbb{R}^{n_q\times n_r}$.  It was shown that linear analog processing prior to the ADC module increases the high SNR capacity, and the gains are further augmented with the use of analog delay elements.   
    \end{itemize}

   We generalize Scenario II, where $f_a:\mathbb{R}^{n_r}\to \mathbb{R}^{n_q}$ consists of $n_q$ real-valued linear functions of $\mathbf{Y}$, 
    and consider scenarios where $f_a(\cdot)$ consists of $n_q$ real-valued polynomial functions of $\mathbf{Y}$. In particular, we consider
    the following scenarios:
    \begin{itemize}[leftmargin=*]
    \item \textbf{Scenario III:  Polynomial Functions with Arbitrary Degree:}  $\mathcal{F}^{III}_a$ consists of vectors of functions $f_a(\mathbf{Y})=(f_{a,1}(\mathbf{Y}),f_{a,2}(\mathbf{Y}),\cdots,f_{a,n_q}(\mathbf{Y}))$, where $f_{a,i}(\mathbf{Y})\in \mathcal{P}(\mathbb{R}^{n_r}), i\in [n_q]$, and $\mathcal{P}(\mathbb{R}^{n_r})$ is the space of  all finite-degree polynomials from $\mathbb{R}^{n_r}$ to $\mathbb{R}$. That is:
    \[f_{a,i}(\mathbf{Y})= \sum_{\substack{ k_1,k_2,\cdots,k_{n_r}\\ k_1+k_2+\cdots +k_{n_r}\leq t}} b_{k^{n_r},i}\prod_{j=1}^{n_r}Y^{k_j}(j), \]
    where $b_{k^{n_r},i}\in \mathbb{R}, (k_1,k_2,\cdots,k_{n_r})\in \mathbb{R}^{n_r}, i\in [n_q], t\in \mathbb{N}$.
    \item \textbf{Scenario IV:  Polynomial Functions with Bounded Degree:}  $\mathcal{F}^{IV}_a$ consists of vectors of functions $f_a(\mathbf{Y})=(f_{a,1}(\mathbf{Y}),f_{a,2}(\mathbf{Y}),\cdots,f_{a,n_q}(\mathbf{Y}))$, where $f_{a,i}(\mathbf{Y})\in \mathcal{P}_d(\mathbb{R}^{n_r}), i\in [n_q]$, and $\mathcal{P}_d(\mathbb{R}^{n_r})$ is the space of  all polynomials with degree at most $d\in \mathbb{N}$.
    \item \textbf{Scenario V:  Implementable Quadratic Functions:}   $\mathcal{F}^{V}_a$ consists of vectors of functions $f_a(\mathbf{Y})=(f_{a,1}(\mathbf{Y}),f_{a,2}(\mathbf{Y}),\cdots,f_{a,n_q}(\mathbf{Y}))$, where $f_{a,i}(\mathbf{Y})$ is generated by the functions $(Y_1,Y_2,\cdots,$ $Y_{n_r}, Y_1^2+Y_2^2+\cdots+Y_{n_r}^2)$. That is,
    \[f_{a,i}(\mathbf{Y})= \sum_{k=1}^{n_r}a_{k,i}Y(k)+ a_{n_r+1,i}\sum_{k=1}^{n_r}Y^2(k), \]
    where $a_{k,i}\in \mathbb{R}, k\in [n_q], i\in [n_r+1]$.
\end{itemize}

Scenario III is an ideal scenario, where analog circuits can be used to generate any arbitrary polynomial function. This scenario is investigated in Section \ref{sec:sen:III}, where the achievable rate region is characterized, and a computable inner bound is provided. However, we argue that this ideal scenario cannot be implemented in practice due to the limitations of analog circuitry which prohibits implementation of high degree polynomial functions. Scenario IV limits the degree of the polynomial function by an integer $d\in \mathbb{N}$. This scenario is investigated in Section \ref{sec:sen:IV}, where the high SNR achievable region is derived, and it is shown that under specific conditions on the number of available one-bit ADCs, the achievable region approaches that of scenario III, i.e. optimal achievable rate.  Scenario V is a special case of Scenario IV, where the function is restricted to specific quadratic polynomials. The achievable rate region for this scenario is characterized for all SNRs in Section \ref{sec:sen:V}. It can be noted that $\mathcal{F}^{I}_a\subset\mathcal{F}^{II}_a\subset\mathcal{F}^{V}_a\subset \mathcal{F}^{IV}_a\subset \mathcal{F}^{III}_a$.

\section{Communication Strategies and Achievable Rates}
In this section, we consider Scenarios III, IV, and V described in Section \ref{sec:form}, and derive the achievable rate region in each case under specific assumptions on the number of available one-bit ADCs $n_q$, and the channel SNR.
\subsection{Scenario III: Polynomials with Arbitrary Degree}
\label{sec:sen:III}
Using the Stone-Weierstrass Theorem on uniform approximation of continuous functions over compact sets (e.g.\cite{de1959stone}), we show that the decoding operation in this scenario is equivalent to a two-step decoding process, where i) the channel output $Y(i,1\!\!:\!\!n_r), i\in [n]$ is discretized with arbitrary discretization bins, so that $\widehat{W}(i,1\!\!:\!\!n_q)= Q'(Y(i,1\!\!:\!\!n_r))$, where $Q':\mathbb{R}^{n_r}\to \{0,1\}^{n_q}$ is an arbitrary function, 
and ii) the discretization indices $\widehat{W}(1\!\!:\!\!n,1\!\!:\!\!n_q)$ are processed jointly to reconstruct the message $\widehat{M}=f_d(\widehat{W}(1\!\!:\!\!n,1\!\!:\!\!n_q))$. This is formalized as follows. 
\begin{Theorem}
\label{th:1}
Let $P>0$, $n_t,n_r,n_q\in \mathbb{N}$, $\mathbf{h}\in \mathbb{R}^{n_t\times n_r}$, and  $X^{n_t}$ be defined on the  probability space  $(\mathbb{R}^{n_t},\mathbb{B}^{n_t},P_{X^{n_t}})$ satisfying the average power constraint $\frac{1}{n_t}\sum_{i=1}^{n_t}\mathbb{E}(X(i))\leq P$.
Then:
\begin{align*}
    C_Q(n_t,n_r,\mathbf{h},P,n_q,\mathcal{F}^{III}_a)\geq  \sup_{\mathsf{A}\in \mathscr{A}_{Rank(\mathbf{h}),n_q}} I_{\mathsf{A}}(X^{n_t};V),
\end{align*}
where $Rank(\mathbf{h})$ is the rank of matrix $\mathbf{h}$, $\mathscr{A}_{Rank(\mathbf{h}),n_q}$ is the set of all possible partitions of $\mathbb{R}^{Rank(\mathbf{h})}$ into $2^{n_q}$ connected regions, $V$ is defined on $[2^{n_q}]$, and the mutual information $I_{\mathsf{A}}(X^{n_t};V)$ is evaluated with respect to $P_{\mathsf{A}}(X^{n_t},V)$ such that:
\begin{align*}
   P_{X^{n_t},V}(\mathcal{C},k)=P(X^{n_t}\in \mathcal{C}, Y\in \mathcal{A}_{k}), \mathcal{C}\in \mathbb{B}^{n_t}, k\in [2^{n_q}]
\end{align*}
where $\mathcal{A}_k$ is the $k$th partition element in $\mathsf{A}$.
\end{Theorem}
\begin{proof}
Please refer to Appendix \ref{App:th:1}.
\end{proof}

\begin{Remark}
Using lower semi-continuity of mutual information, and data processing inequality \cite{Pinsker}, it follows that as $n_q\to \infty$, the capacity approaches that of continuous-output Gaussian channels. That is,
$ \lim_{n_q\to \infty}C(n_t,n_r,\mathbf{h},P,n_q,\mathcal{F}^{III}_a)$ is equal to the capacity of the continuous-output Gaussian channel $\mathbf{Y}=\mathbf{h}\mathbf{X}+\mathbf{N}$.
\end{Remark}

Theorem \ref{th:1} provides a lower-bound on the capacity, however, this bound is not necessarily computable since it requires optimization over all partitions in  $\mathscr{A}_{Rank(\mathbf{h}),n_q}$. The following theorem provides a computable inner-bound to the one given in Theorem \ref{th:1}. The theorem uses the singular value decomposition (SVD) in the analog domain to transform the communication system into $s=Rank(\mathbf{h})$ parallel, non-interfering channels.  The $i$th parallel channel is allocated a number of $n_{q,i}$ one-bit ADCs, where $\sum_{i\in [s]}n_{q,i}=n_q$, and is allocated $\frac{P_i}{P}$ fraction of the power budget, where $\sum_{i\in [s]}P_i=P$.

\begin{Theorem}
\label{th:2}
Let $P>0$, $n_t,n_r,n_q\in \mathbb{N}$, $\mathbf{h}\in \mathbb{R}^{n_t\times n_r}$
Then:
\begin{align}
   & \label{eq:th:2} C_Q(n_t,n_r,\mathbf{h},P,n_q,\mathcal{F}^{III}_a)\geq 
   \\&\qquad \qquad \qquad  \max_{(n_{q,i}, i\in [s])\in \mathcal{N}}\max_{\substack{(P_i, i\in [s])\in \mathcal{P}}}\sup_{\mathsf{A_i}\in \mathscr{A}_{1,n_{q,i}}}\sup_{P_{\widetilde{X}^s}} \sum_{k =1}^{s} I_{\mathsf{A}_i}(\widetilde{X}_k;V_k),\nonumber
\end{align}
where  $\mathcal{N}\triangleq\{(n_{q,i},i\in [s]:\sum_{i\in [s]}n_{q,i}=n_q\}$,  $\mathcal{P}\triangleq\{(P_i,i\in [s]:\sum_{i\in[s]}P_{i}=P\}$, $\widetilde{Y}_k = \sigma_{k} \widetilde{X}_k+N_k$, $N^s$ is a vector of i.i.d. zero-mean Gaussian variables with unit variance, $\sigma_k$ is the $k$th eigenvalue of $\mathbf{h}$, and the mutual information is evaluated with respect to $P^{\mathsf{A}_i}_{X^{n_t},V_k}$ such that:
\begin{align*}
   P^{\mathsf{A}_i}_{\widetilde{X}_k,V_k}(\mathcal{C} ,\ell)=P(X^{n_t}\!\in\! \mathcal{C}, Y^{n_r}\!\in\! \mathcal{A}_{\ell}), \mathcal{C}\!\in\! \mathbb{B}^{n_t}, k\in [2^{n_i}], \ell \in [2^{n_{q_i}}].
\end{align*}
\end{Theorem}
The proof follows by similar arguments as Theorem 3 of \cite{khalili2021mimo} along with proof of Theorem \ref{th:1} and is omitted for brevity.
\begin{Remark}
The first two maximizations in Equation \eqref{eq:th:2} are over finite sets, and the two supremums are taken over all one dimensional partitions along with the distribution of $X$. The latter optimization has been studied extensively in the literature, and it was shown that the maximum rate is achieved by putting the mass of $P_X$ on a finite number of at most $2^{n_q}$ points \cite{witsenhausen1980some,singh2009limits}.
\end{Remark}
\subsection{Scenario IV: Polynomials with Bounded Degree}
\label{sec:sen:IV}
In this scenario, we consider a special case of Scenario IV, where the analog functions are restricted to specific quadratic polynomials. In particular, the function is generated by $\{x_1,x_2,\cdots,x_{Rank(\mathbf{h})},\sum_{i=1}^{Rank(\mathbf{h})}x_i^2\}$. We have the following theorem for the high SNR capacity.
\begin{Theorem}
Let $n_t,n_r,n_q\in \mathbb{N}$, $\mathbf{h}\in \mathbb{R}^{n_t\times n_r}$ and let the maximum polynomial degree in $\mathcal{F}^{IV}$ be  $d\in \mathbb{N}$. Then,
\begin{align*}
   & \lim_{P\to \infty} C_Q(n_t,n_r,\mathbf{h},P,n_q,\mathcal{F}^{IV}_a)\geq 
 \max\left(n_q,\log{Rank(\mathbf{h})+d\choose d}\right),\nonumber
\end{align*}
\end{Theorem}

\textit{Proof Outline.} Following the arguments in Theorem 1 in \cite{khalili2021mimo}, the maximum number of messages which can be transmitted reliably at high SNR is equal to the maximum number of distinct quantization regions which can be produced in $\mathbb{R}^{Rank{(\mathbf{h})}}$ using the analog processing and ADC operation (as long as it is not greater than $2^{n_q}$). We argue that this number is greater than or equal to the  Vapnik-Chervonenkis Dimension ($VCdim$) of $\mathcal{F}^{IV}$. The proof is completed by noting that $VCdim(\mathcal{F}^{IV})= {Rank(\mathbf{h})+d\choose d}$ as shown in \cite{anthony1995classification}. To see the former statement, let $\mathbf{x}_1,\mathbf{x}_2,\cdots, \mathbf{x}_{\ell}\in \mathbb{R}^{Rank(\mathbf{h})}$ be a set of points which are shattered by  $\mathcal{F}^{IV}$, where $\ell=VCdim(\mathcal{F}^{IV})$. Then, by definition of $VCdim$, there exist polynomial discriminants $f_{a,j}(\cdot), j\in [2^{\ell}]$ in $\mathcal{F}^{IV}$ such that $(Q_0(f_{a,j}(\mathbf{x}_1)),Q_0(f_{a,j}(\mathbf{x}_2)),\cdots, Q_0(f_{a,j}(\mathbf{x}_\ell))$ is the binary representation of $j$. Let $f_{a,j_1},f_{a,j_2},\cdots,f_{a,j_{n_q}}$ be such that $sign(f_{a,j_k}(\mathbf{x}_{t}))=(-1)^{mod_{2^k}(t)}, k\in [n_q], t\in [\ell]$. Then, $(Q_0(f_{a,j_1}(\mathbf{x}_t)),Q_0(f_{a,j_2}(\mathbf{x}_t)),\cdots,Q_0(f_{a,j_{n_q}}(\mathbf{x}_t)))$ is the binary representation of $t$. As a result, if $\{\mathbf{x}_1,\mathbf{x}_2,\cdots,\mathbf{x}_\ell\}$ is taken to be the channel input alphabet, and the functions $f_{a,j_1},f_{a,j_2},\cdots,f_{a,j_{n_q}}$ are used as the analog processing functions, the receiver can reconstruct the index $T$ of the transmitted symbol $\mathbf{x}_T$ without error by finding its binary representation as described above. This completes the proof. 

\begin{Remark}
    In the low-SNR regime, one could potentially extend Theorems \ref{th:1} and \ref{th:2} to the bounded polynomial degree scenario considered here by evaluating the error bounds on approximation of
functions using the Bernstein Polynomial (e.g. \cite{sukkrasanti2008error}). 
\end{Remark}
\subsection{Scenario V: Subset of Quadratic Functions}
\label{sec:sen:V}
In this scenario, we consider a special case of Scenario IV, where the analog functions are restricted to specific quadratic polynomials. In particular, the function is generated by $\{x_1,x_2,\cdots,x_{Rank(\mathbf{h})},\sum_{i=1}^{Rank(\mathbf{h})}x_i^2\}$. We have the following theorem for the high SNR capacity.
\begin{Theorem}
\label{th:4}
Let $n_t,n_r,n_q\in \mathbb{N}$, $\mathbf{h}\in \mathbb{R}^{n_t\times n_r}$. Then,
\begin{align*}
   & \lim_{P\to \infty} C_Q(n_t,n_r,\mathbf{h},P,n_q,\mathcal{F}^{V}_a)= \log\left(\sum_{i=0}^{Rank(\mathbf{h})+1} {n_q \choose i}- {n_q-1\choose Rank(\mathbf{h})}\right),\nonumber
\end{align*}
\end{Theorem}
\textit{Proof Outline.}
The proof follows by noting that each polynomial in  $\mathcal{F}^{IV}$ can be mapped to a hyperplane in the $\mathbb{R}^{Rank(\mathbf{h})+1}$ by considering the mapping $(x_1,x_2,\cdots,x_{Rank(\mathbf{h})})\mapsto (x_1,x_2,$ $\cdots,x_{Rank(\mathbf{h})},\sum_{i=1}^{Rank(\mathbf{h})}x_i^2)$. Furthermore, it is well-known that the number of partition regions in $\mathbb{R}^{Rank(\mathbf{h})+1}$ generated by $n_q$ hyperplanes is $\sum_{i=0}^{Rank(\mathbf{h})+1} {n_q \choose i}$. Also, due to the convexity of the $Rank(\mathbf{h})$-dimensional surface  $\mathcal{L}=\{(x_1,x_2,\cdots,x_{Rank(\mathbf{h})},\sum_{i=1}^{Rank(\mathbf{h})}x_i^2)| x_i\in \mathbb{R}\}$ for every closed and convex partition region (polyhedra) $\mathcal{A}_j$, either $\mathcal{L}\cap \mathcal{A}_j=\phi$, or there is a region $\mathcal{A}_{j'}$ sharing a vertex with $\mathcal{A}_j$ such that $\mathcal{L}\cap \mathcal{A}_{j'}=\phi$. So, the maximum number of partition regions with which $\mathcal{L}$ intersects is at most $\sum_{i=0}^{Rank(\mathbf{h})+1} {n_q \choose i}- \beta_{Rank(\mathbf{h})+1,n_q}$, where $\beta_{Rank(\mathbf{h})+1,n_q}$ is the number of closed (bounded) partition regions. Furthermore, it can be shown that by scaling the partition appropriately, one can ensure that all closed partition regions lie inside $\mathcal{L}$ so that the maximum number of intersecting regions $\sum_{i=0}^{Rank(\mathbf{h})+1} {n_q \choose i}- \beta_{Rank(\mathbf{h})+1,n_q}$ is achieved. The proof is completed by noting that $ \beta_{Rank(\mathbf{h})+1,n_q}={n_q-1\choose Rank(\mathbf{h})}$ as shown in \cite{ho2006number}.
\begin{Remark}
The capacity region derived in Theorem \ref{th:4} can be equivalently stated as  $C(n_t,n_r,\mathbf{h},P,n_q,$ $\mathcal{F}^{V}_a)= \log(\alpha_{Rank(\mathbf{h})+1,n_q})$, where $\alpha_{Rank(\mathbf{h})+1,n_q)}$ is the maximum number of distinct regions generated by $n_q$ hyperplanes passing through the origin in  $\mathbb{R}^{Rank(\mathbf{h})+1}$, and $\alpha_{Rank(\mathbf{h})+1,n_q}=2\sum_{i=0}^{Rank(\mathbf{h})} {n_q-1 \choose i}$. The equality of these two formulas was shown in \cite{ho2006number}.
\end{Remark}
Similar to Theorem \ref{th:2}, one could use SVD to derive the following lower-bound on the low SNR capacity. 
\begin{Theorem}
\label{th:5}
Let $P>0$, $n_t,n_r,n_q\in \mathbb{N}$ $\mathbf{h}\in \mathbb{R}^{n_t\times n_r}$
Then:
\begin{align}
   & \label{eq:th:5} C_Q(n_t,n_r,\mathbf{h},P,n_q,\mathcal{F}^{V}_a)\geq 
   \\&\qquad \qquad \qquad  \max_{(n_{q,i}, i\in [s])\in \mathcal{N}}\max_{\substack{(P_i, i\in [s])\in \mathcal{P}}}\sup_{\mathsf{A_i}\in \mathscr{B}_{n_{q,i}}}\sup_{P_{\widetilde{X}^{n_t}}} \sum_{k =1}^{s} I_{\mathsf{A}_i}(\widetilde{X}_k;V_k),\nonumber
\end{align}
where  $\mathcal{N}\triangleq\{(n_{q,i},i\in [s]:\sum_{i\in [s]}n_{q,i}=n_q\}$,  $\mathcal{P}\triangleq\{(P_i,i\in [s]:\sum_{i\in[s]}P_{i}=P\}$, $\widetilde{Y}_k = \sigma_{k} \widetilde{X}_k+N_k$, $N^s$ is a vector of i.i.d. zero-mean Gaussian variables with unit variance, $\sigma_k$ is the $k$th eigenvalue of $\mathbf{h}$, $\mathscr{B}_{n_{q_i}}$ is the set of all partitions of $\mathbb{R}$ into $\zeta$  intervals, where $\zeta=2n_{q_i}$ otherwise, and the mutual information is evaluated with respect to 
$P^{\mathsf{A}_i}_{\widetilde{X}^{n_t},V_k}$ such that:
\begin{align*}
   P^{\mathsf{A}_i}_{\widetilde{X}^{n_t},V_k}(\mathcal{C} ,\ell)=P(X^{n_t}\!\in\! \mathcal{C}, Y^{n_r}\!\in\! \mathcal{A}_{\ell}), \mathcal{C}\in \mathbb{B}^{n_t}, k\in [2^{n_{q_i}}], \ell\in [|\mathsf{A}_i|].
\end{align*}
\end{Theorem}
\begin{Remark}
It can be noted that if the analog processing function is restricted to linear functions as in \cite{molisch2017hybrid,heath2016overview,alkhateeb2014mimo,nossek2006capacity,abbasISIT2018,rini2017general,mezghani2009transmit,khalili2020throughput,dutta2020capacity},  $\mathscr{B}_{n_{q_i}}$ in Theorem \ref{th:5} would be replaced by the set of all partitions of $\mathbb{R}$ into $n_{q_i}+1$ partitions which leads to a strictly smaller achievable rate.
\end{Remark}
\section{Circuit Design for nonlinear Analog Operations}
\label{sec:cir}
In the prequel, we have evaluated the fundamental limits of communication, in terms of achievable rates in MIMO systems with one-bit ADCs equipped with nonlinear analog operations (limited degree polynomials) prior to the ADC operation. In this section, we provide an example of a circuit  and provide circuit simulations to justify the feasibility of such nonlinear operations. 

The implementation relies on the fact that Complementary Metal-Oxide-Semiconductor (CMOS) and Bipolar transistors --- the core components of integrated analog circuits --- manifest inherent device-centric nonlinearity by generating integer harmonic frequencies when excited by a sinusoidal input waveform, i.e. $ \cos (\omega t+ \phi)$. Various mathematical models to capture transistor nonlineariy exist, among which the adoption of Volterra-Weiner series shown in \cite{JSSC2017} addresses the general scenario. The nonlinear response of a transistor to a sinusoidal input waveform depends on excitation frequency, $\omega$. 
The received modulated signal in a MIMO receiver is a severely attenuated version of the input signal, and is non-monotone. Consequently, it may not be directly applied to a nonlinear transistor, and a two-step procedure involving a pre-processing step followed by nonlinear analog operations is required as described below.  
\\\textbf{Step 1: Conversion of the received signal into a monotone sinusoidal.} 
To explain this step, let us assume that the non-monotone received signal is $A sinc(\omega_0 t), A\in \{-2,-1,1,2\}, \omega_0>0$. This is injected into an integrator circuit [\textit{cf.} Fig. \ref{fig:sinc}(a)]. For each channel-use, the integrator output after $T_s>0$ seconds (extracted using switches $SW_1$ and $SW_2$) is injected to control the voltage of complementary switches $SW_3$ and $SW_4$, and the two identical  variable capacitors with opposite polarities in Fig. \ref{fig:sinc}(b). The resonator circuit generates monotone sinusoidal waveforms with amplitudes proportional to $|INT(A)|$. The dependence on $|INT(A)|$ is due to the fact that the associated quality factor $Q$ of the variable capacitors changes linearly within the possible range of $|INT(A)|$, thus generating sinusoidal waveforms with varying amplitude and frequency [\textit{cf.} Fig. \ref{fig:sinc}(b)].
\\\textbf{Step 2: Generating polynomial outputs.}
We apply the sinusoidal waveforms generated by the resonator to a nonlinear circuit, so that the frequency harmonics are generated at the output, with amplitude of \textit{i}th harmonic proportional to $B^i$, where $B$ is the amplitude of the sinusoidal input.  To elaborate, in Figure \ref{fig:sinc}(c), we incorporate a differential amplifier circuit \cite{razavi2005design}.  Based on \cite{JSSC2017}, for this circuit, we have $
B_{out}=\alpha B^2$,    
where $\alpha$ is the coefficient of the Volterra-Weiner representation of the nonlinear circuit \cite{volterra}, and $B_{out}$ is the amplitude of the generated second harmonic signal. It can be noted that the value of $\alpha$ within the resonator output frequency range remains constant. The amplitude ratio of generated second harmonic waveforms in Fig.\ref{fig:sinc}(c) and the fundamental frequency components in Fig. \ref{fig:sinc}(b) illustrate the feasibility of producing polynomial functions of the input amplitude in the analog domain. 
\begin{figure}[t!]
\centering
\includegraphics[width=0.9\linewidth]{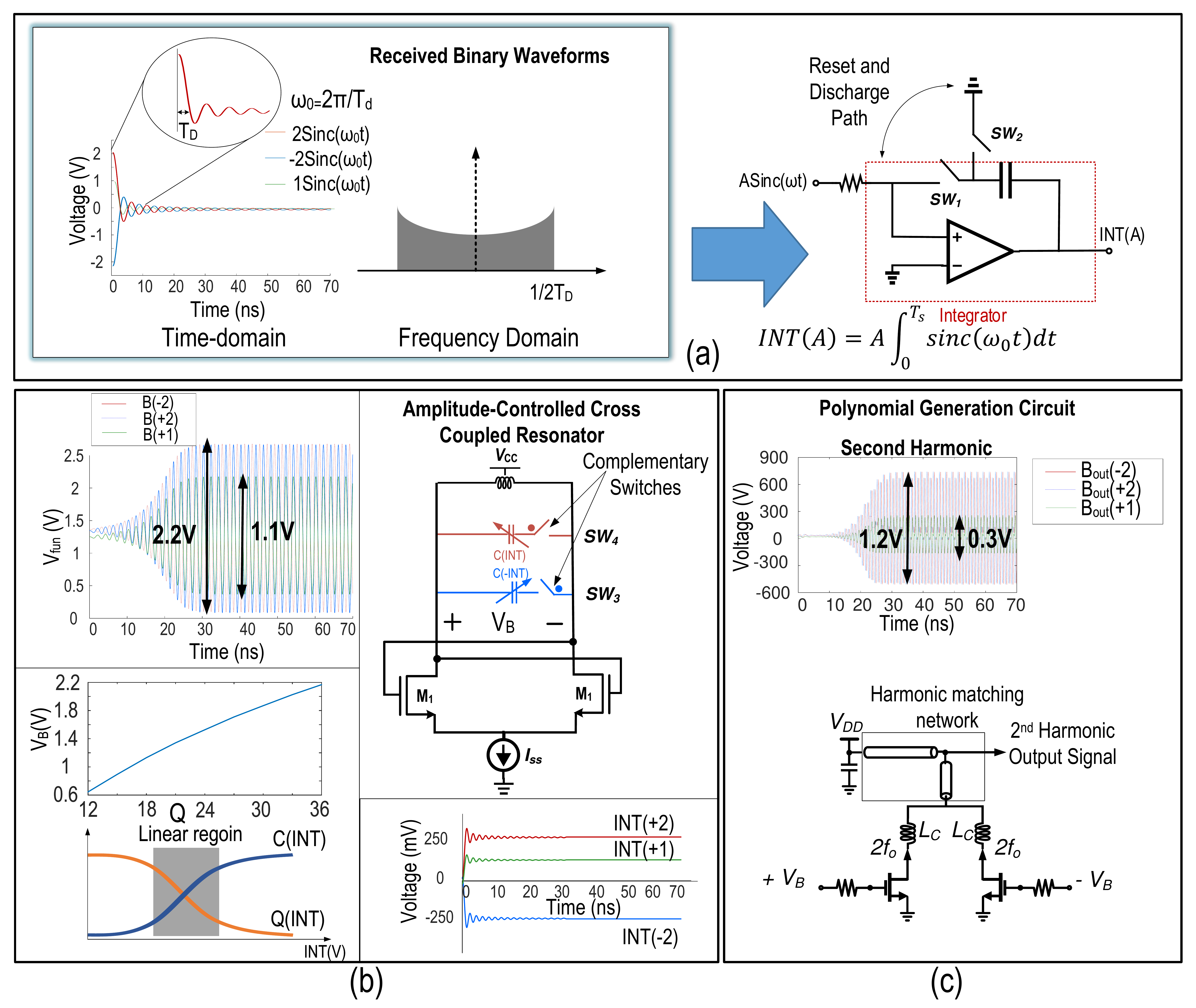}
\caption{(a) Passing received waveform through an integrator, (b) resonator amplitude follows integrator output, (c) polynomial harmonic generator.}
\label{fig:sinc}
\end{figure}

\section{Conclusion}
The application of nonlinear analog operations in MIMO receivers was considered. A receiver architecture consisting of linear analog combiners, implementable nonlinear analog operators, and one-bit threshold ADCs was designed, and
the fundamental information theoretic performance limits of the resulting communication system were investigated. To justify the feasibility of the nonlinear operations, an analog circuit was introduced, and circuit simulations exhibiting the generation of nonlinear analog operations were provided.
\begin{appendices}
\section{Proof of Theorem \ref{th:1}}
\label{App:th:1}
We provide an outline of the proof for $Rank(\mathbf{h})=n_r$. The proof for $Rank(\mathbf{h})\leq n_r$ follows by similar arguments. Fix partition $\mathsf{A}\in \mathscr{A}_{n_r,n_q}$. 
Define the collection of functions
\begin{align*}
    f_j(y^{n_r})= (-1)^{mod_{2^j}(k)}||y^{n_r}-\partial\mathcal{A}_k||_2,
\end{align*}
where $y^{n_r}\in \mathbb{R}^{n_r}$, $j\in \{0,1,\cdots,n_q-1\}$, $mod_b (a)$ denotes $a$ modulo $b$, $k\in [2^{n_q}]$ is the index of the partition region for which $y^{n_r}\in \mathcal{A}_k$,  and  $||y^{n_r}-\partial\mathcal{A}_k||_2$ is the $\ell_2$ distance between $y^{n_r}$ and the boundary of the region $\mathcal{A}_k$. The function $f_j(\cdot)$ is continuous and its roots are the boundary points of the partition regions $\mathcal{A}_{k'}, k'\in [2^{n_q}]$. Furthermore, its value is positive for all interior points of regions $\mathcal{A}_{k'}, k'\in [2^{n_q}]$ for which $mod_{2^j}k'$ is even and is negative otherwise. As a result, $Sign(f_j(y^{n_r})), j\in \{0,1,\cdots,n_q-1\}$ is the binary representation of the index of $\mathcal{A}_k$, where $y^{n_r}\in \mathcal{A}_k$, and $Sign(\cdot)$ is equivalent to a zero-threshold one-bit ADC. So, $I(X^{n_t};Q_0^{n_q}(W^{n_q}))=I_{\mathsf{A}}(X^{n_t};V)$, where $Q_0^{n_q}(\cdot)$ represents $n_q$ zero-threshold one-bit ADCs and $W_i=f_{i-1}(Y^{n_r}), i\in [n_q]$.
It remains to show that each $f_j(\cdot)$ is `well-approximated' by a polynomial function of $y^{n_r}$. To see this, we let $L>0$ and define $E_i\triangleq\mathbbm{1}(|Y_i|<n_r L), i\in [n_r]$. We note that:
\begin{align*}
&I_{\mathsf{A}}(X^{n_t}; V) \leq \sum_{i=1}^{n_r}H(E_i)+I(X^{n_t};V|E^{n_r})
\\&\stackrel{(a)}{\leq}  \sum_{i=1}^{n_r}H(E_i)+
P(\exists i\in [n_r]: E_i=0)\log{2^{n_q}} +
I_{\mathsf{A}}(X^{n_t};V|E_i=1, i\in [n_r])  
\\&\leq  \sum_{i=1}^{n_r}H(E_i)+ \sum_{i=1}^{n_r}P(|Y_i|>n_r L)n_q+I_{\mathsf{A}}(X^{n_t};V|E_i=1, i\in [n_r])
\\& \stackrel{(b)}{\leq} 
 \sum_{i=1}^{n_r}H(E_i)+\frac{\gamma_Y  n_q}{L}+I_{\mathsf{A}}(X^{n_t};V|E_i=1, i\in [n_r]),
\end{align*}
where we have defined $\gamma_Y\triangleq \frac{1}{n_r}\sum_{i=1}^{n_r}\mathbb{E}(|Y_i|)$,  (a) holds since $V$ takes at most $2^{n_q}$ values
, and (b)
follows from Markov's inequality.  Note that $\gamma_Y< \infty$ since $\frac{1}{n_t}\sum_{i=1}^{n_t}\mathbb{E}(X^2_i)\leq P <\infty$. Consequently, by the lower-semi continuity of mutual information $I(X^{n_t}; V)$ approaches $I_{\mathsf{A}}(X^{n_t};V|E_i=1, i\in [n_r])$ as $L\to \infty$. On the other hand, note that $[-n_rL,n_rL]^{n_r}$ equipped with the $\ell_2$ distance is a compact metric space. Hence,  
by the Stone-Weierstrass Theorem, the polynomial functions are dense in functions defined on $[-n_rL,n_rL]^{n_r}$ and there exists a sequence of polynomial functions $f_{t,j}(\cdot), t\in \mathbb{N}$ which converge uniformly to the
restriction of $f_j(\cdot), j\in \{0,1,2,\cdots,n_q-1\}$ to $[-n_rL,n_rL]^{n_r}$. Consequently, for an arbitrary $\epsilon>0$, there exists $L$ and $t$ large enough, so that $\sum_{i=1}^{n_r}H(E_i)+\frac{\gamma_Yn_q}{L}\leq \epsilon$, $P(\widehat{W}^{n_q}=Q_0^{n_q}(W'^{n_q}))\geq 1-\epsilon$, where $W'_i\triangleq f_{t,j}(Y^{n_r})$,
and  \[|I_{\mathsf{A}}(X^{n_t};V|E_i\!=\!1, i\!\in\! [n_r])-I(X^{n_t};Q_0^{n_q}(W'^{n_q})|E_i\!=\!1, i\!\in \![n_r])|\leq \epsilon.\]

\end{appendices}

 \clearpage
\bibliographystyle{unsrt}
\bibliography{References}

\end{document}